\documentclass[10pt]{article}

\usepackage{amsmath}
\usepackage{amssymb, amscd, amsthm}
\usepackage[all]{xy}
\usepackage[dvips]{graphicx}
\usepackage{verbatim}
\usepackage[perpage,symbol*]{footmisc}
\newtheorem{theorem}{Theorem}
\newtheorem{definition}{\textbf{Definition}}

\newtheorem{lemma}{\textbf{Lemma}}
\newtheorem{remark}{\textbf{Remark}}

\newtheorem{example}{\textbf{Example}}

\usepackage[justification=centering]{caption}
\usepackage{longtable}
\usepackage{multirow}

\usepackage{amsfonts}
\usepackage{graphicx}
\usepackage{color}
\usepackage{url}
\usepackage{longtable}
\usepackage{booktabs}
\usepackage{float}
\restylefloat{table}

\newcommand{\tabincell}[2]{\begin{tabular}{@{}#1@{}}#2\end{tabular}}

\newcommand{\Tr}{{{\rm Tr}}}

\begin{document}
\title{Differential uniformity and costacyclic code from some power mapping}

\author{Yuehui Cui,\mbox{ } Jinquan Luo}
\date{School of Mathematics and Statistics,
 Central China Normal University, Wuhan 430079, China}
\maketitle
\insert\footins{\small{\it Email addresses} :luojinquan@mail.ccnu.edu.cn (Jinquan Luo); hfcyh1@163.com (Yuehui Cui).}
{\centering\section*{Abstract}}
 \addcontentsline{toc}{section}{\protect Abstract}
 \setcounter{equation}{0}

In this paper, we study the differential properties of   $x^d$ over $\mathbb{F}_{p^n}$ with $d=p^{2l}-p^{l}+1$. By studying the differential equation of $x^d$ and the number of rational points on some curves over finite fields, we completely determine differential spectrum of $x^{d}$. Then we investigate the $c$-differential uniformity of $x^{d}$. We also calculate the value distribution of a class of exponential sum related to $x^d$. In addition, we obtain a class of six-weight consta-cyclic codes, whose weight distribution is explicitly determined. Part of our results is a complement of the works shown in [\ref{H1}, \ref{H2}] which mainly focus on cross correlations.

\medskip
\noindent{\large\bf Keywords: }\medskip Differential spectrum, $c$-differential uniformity, Constacyclic code, Exponential sum, Weight distribution.

\noindent{\bf2010 Mathematics Subject Classification}: 94A60, 11T06, 11T23.

\section{Introduction}

Let $\mathbb{F}_{p^n}$ be the finite field with $p^n$ elements, where $p$ is a prime and $n$ is a positive integer. Let $\mathbb{F}_{p^n}^*$ denote the multiplicative group of the finite field $\mathbb{F}_{p^n}$. For any function $f:\mathbb{F}_{p^n}\rightarrow\mathbb{F}_{p^n}$, the derivative function of $f$ in respect to any element $a \in \mathbb{F}_{p^n}$  is a function from $\mathbb{F}_{p^n}$ to $\mathbb{F}_{p^n}$  defined by
\begin{equation*}
  D_af(x)=f(x+a)-f(x), \forall x \in \mathbb{F}_{p^n}.
\end{equation*}
For any $a,b \in \mathbb{F}_{p^n}$, let
\begin{equation*}
  \delta(a,b)=\#\left\{x \in \mathbb{F}_{p^n} \mid D_af(x)=b \right\}.
\end{equation*}
The differential uniformity of $f$ is defined as
\begin{equation*}
  \delta=\max\left\{\delta(a,b) \mid a \in \mathbb{F}_{p^n}^*, b \in \mathbb{F}_{p^n}  \right\}.
\end{equation*}

If $\delta=1$, then $f$ is called a perfect nonlinear (PN) function. If $\delta=2$, then $f$ is called an almost perfect nonlinear (APN) function. The readers can refer to [\ref{APN1}, \ref{APN2}, \ref{book cc}, \ref{quad APN}] for more information on PN and APN functions. In recent years, the design of Substitution box (S-box) that can effectively resist differential attacks has become a hot topic in cryptographic research. Studies have shown that the power functions with low differential uniformity offer significant advantages in S-box design. Therefore, researching the power functions with low differential uniformity and exploring their application in S-box design has become an important research direction in the field of cryptography. For related details, we refer to [\ref{sbox1}, \ref{sbox2}]. When $f(x)=x^d$ is a power mapping,
\begin{equation*}
  (x+a)^d-x^d=b \Leftrightarrow a^d \left (\left(\frac{x}{a}+1\right)^d-\left(\frac{x}{a}\right)^d \right)=b
\end{equation*}
implying that $\delta(a,b)=\delta(1,\frac{b}{a^d})$ for all $a \in \mathbb{F}_{p^n}^*$ and $b \in \mathbb{F}_{p^n}$. Thus, the differential properties of $f(x)=x^d$ are completely determined by the values of $\delta(1,b)$ when $b$ runs through $\mathbb{F}_{p^n}$.

\begin{definition}
  Let $f(x)=x^d$ be a power function from $\mathbb{F}_{p^n}$ to $\mathbb{F}_{p^n}$ with differential uniformity $\delta$. Denote
\begin{equation*}
  \omega_i=\#\left\{b \in \mathbb{F}_{p^n} \mid \delta(1,b)=i \right\},
\end{equation*}
where $0\leq i\leq\delta$. The differential spectrum of $f(x)=x^d$ is defined as
\begin{equation*}
  \mathbb{DS}=\left\{\omega_i \mid 0\leq i\leq\delta \right\}.
\end{equation*}
\end{definition}
Sometimes we ignore the zeros in $\mathbb{DS}$. We also have two basic identities [\ref{first def}].
\begin{equation}\label{two basic ide}
  \sum\limits_{i=0}^{\delta}\omega_i=p^n \mbox{ and } \sum\limits_{i=0}^{\delta}i\omega_i=p^n.
\end{equation}
These two identities are useful in calculating the differential spectrum of $f$. However, determining the differential spectrum of a power mapping remains theoretically challenging. We list some results in table 1. Differential spectrum of a power mapping has many applications, such as cross-correlation distribution, Walsh transformation, $t$-designs and coding theory. Readers can refer to [\ref{ds and cc}, \ref{tdesigns}-\ref{app code2}, \ref{walsh cc}] for more information.

\begin{table}[!ht]
  \centering
  \caption{Some power functions $f(x)=x^d$ over $\mathbb{F}_{p^n}$ with known differential spectrum.\\}
  \label{table1}
  \begin{tabular}{llcc}
    \toprule
    $d$                                      & Conditions                                                                       & $\delta$ & References           \\
    \hline
$2^t+1$  &  $\gcd(t,n)=s$  & $2^s$ & [\ref{first def},\ref{tab 1}] \\ \hline

$2^n-2$  & $n\geq2$  & $2$ or $4$ & [\ref{first def},\ref{tab 1}] \\ \hline

$2^{2t}-2^t+1$  &  $\gcd(t,n)=s$, $\frac{n}{s}$ odd  & $2^s$ & [\ref{first def}] \\ \hline

$2^{2k}+2^k+1$  &  $n=4k$  & $4$ & [\ref{first def},\ref{tab 2}] \\ \hline

$2^t-1$  &  $t=3, n-2$  & $6$ or $8$ & [\ref{tab 3}] \\ \hline

$2^t-1$  &  $t=n/2, n/2+1$, $n$ even  & $2^{\frac{n}{2}}-2$ or $2^{\frac{n}{2}}$ & [\ref{first def}] \\ \hline

$2^t-1$  &  $t=(n-1)/2, (n+3)/2$, $n$ odd  & $6$ or $8$ & [\ref{tab 4}] \\ \hline

$2^{3l}+2^{2l}+2^l-1$  &  $n=4l$  & $2^{2l}$ & [\ref{tab 5}] \\ \hline

$2^m+2^{(m+1)/2}+1$  &  $n=2m$, $m\geq5$ odd  & $8$ & [\ref{tab 6}] \\ \hline

$2^{m+1}+3$  &  $n=2m$, $m\geq5$ odd  & $8$ & [\ref{tab 6}] \\ \hline

$\frac{2^m-1}{2^k+1}+1$  &  $n=2m$, $\gcd(k,m)=1$  & $2^m$ & [\ref{tab 8}] \\ \hline

$2\cdot3^{\frac{n-1}{2}}+1$  &  $n$ odd  & $4$ & [\ref{ds and cc}] \\ \hline

$\frac{p^k+1}{2}$  &  $\gcd(k,n)=e$  & $\frac{p^e-1}{2}$ or $p^e+1$ & [\ref{choi}] \\ \hline

$\frac{p^n+1}{p^m+1}+\frac{p^n-1}{2}$  & $p\equiv 3\,\,({\rm mod}\,\,4)$, $m\mid n$, $n$ odd  & $\frac{p^m+1}{2}$ & [\ref{choi}] \\ \hline

$p^{2k}-p^k+1$  & $p$ odd, $\gcd(k,n)=e$, $\frac{n}{e}$ odd  & $p^e+1$ & [\ref{lei},\ref{yan lei}] \\ \hline

$p^n-3$  &  $p$ odd, any $n$  & $\leq5$ & [\ref{pn-3},\ref{pn-3 2}] \\ \hline

$p^m+2$  &  $n=2m$, $p\geq3$  & $4$ & [\ref{man}] \\ \hline

$\frac{p^n+3}{2}$  &  $p\geq5$, $p^n\equiv 1\,\,({\rm mod}\,\,4)$ & $3$ & [\ref{qulongjiang}] \\ \hline

$\frac{5^n-3}{2}$  &  $n\geq2$ & $4$ or $5$ & [\ref{yan p5}] \\ \hline

$\frac{p^n-3}{2}$  &  $p^n\geq7$, $p^n\equiv 3\,\,({\rm mod}\,\,4)$, $p^n\neq27$ & $2$ or $3$ & [\ref{yan ds}] \\ \hline

$k(p^m-1)$  &  odd $p$, $n=2m$, $\gcd(k,p^m+1)=1$ & $p^m-2$ & [\ref{kpm-1}] \\ \hline

$\frac{p^n+3}{2}$  & $p^n\equiv 3\,\,({\rm mod}\,\,4)$, $p\neq3$ & $2$ or $4$ & [\ref{yan pn+3}] \\ \hline

$2p^m-1$  &  $n=2m$, $p$ odd & $p^m$ & [\ref{yan niho}] \\ \hline

$\frac{p^n+1}{4}+\frac{p^n-1}{2}$  &  $p^n\equiv 3\,\,({\rm mod}\,\,8)$ or $p^n\equiv 7\,\,({\rm mod}\,\,8)$ & $\leq4$ & [\ref{yan tan},\ref{xia jz}] \\ \hline

$\frac{p^n+1}{4}$  &  $p^n\equiv 3\,\,({\rm mod}\,\,8)$ or $p^n\equiv 7\,\,({\rm mod}\,\,8)$ & $\leq4$ & [\ref{yan tan},\ref{xia jz}] \\ \hline

\tabincell{c}{$d(p^k+1)\equiv2$\\(mod $p^n-1)$}  &  $\gcd(n,k)=e$, $\frac{n}{e}$ odd & $\frac{p^e+1}{2},\frac{p^e-1}{2}$ or $p^e+1$ & [\ref{jssms}] \\ \hline

$p^{2l}-p^{l}+1$ & $n=4l$, any prime $p$ & $p^{2l}-p^l$ & This paper \\
  \bottomrule
  \end{tabular}
\end{table}

To define the capacity of a particular function to withstand multiplicative differential attack [\ref{mult diff}], Ellingsen et al. [\ref{C-diff def}] proposed the concept of $c$-differential uniformity, which is described as follows.
\begin{definition}
Let $a,c\in\mathbb{F}_{p^n}$. For any function $f:\mathbb{F}_{p^n}\rightarrow\mathbb{F}_{p^n}$, the (multiplicative) $c$-derivative of $f$ in respect to any element $a$ is defined as
\begin{equation*}
  _cD_af(x)=f(x+a)-cf(x).
\end{equation*}
Let
\begin{equation*}
  _c\Delta_f(a,b)=\#\left\{x \in \mathbb{F}_{p^n} \mid {_cD_af(x)=b} \right\}.
\end{equation*}
We call
\begin{equation*}
  _c\Delta_f=\max\left\{_c\Delta_f(a,b) \mid a,b \in \mathbb{F}_{p^n}\mbox{ and } (a,c)\neq(0,1) \right\},
\end{equation*}
the $c$-differential uniformity of $f$. If $_c\Delta_f=\delta$, then we say $f$ is differentially $(c, \delta)$-uniform.
\end{definition}
If $_c\Delta_f=1$, then $f$ is called a perfect $c$-nonlinear (PcN) function. If $_c\Delta_f=2$, then $f$ is called an almost perfect $c$-nonlinear (APcN) function. Functions exhibiting low c-differential uniformity demonstrate strong resistance to certain forms of differential cryptanalysis. In order to calculate the $c$-differential uniformity of a power function $f(x)=x^d$ over $\mathbb{F}_{p^n}$ with $c\neq1$, the following result was introduced.
\begin{lemma}\label{c delta f}
([\ref{C-diff uni}]) Let $f(x)=x^d$ be a power function over $\mathbb{F}_{p^n}$. Then
\begin{equation*}
  _c\Delta_f=\max \big\{ \{_c\Delta_f(1,b) :b \in \mathbb{F}_{p^n} \}\cup \{\gcd(d,p^n-1) \} \big\}.
\end{equation*}
\end{lemma}
The rest of this paper is organized as follows. In section \ref{Preliminaries}, we introduce some notations and auxiliary tools. In section \ref{ds}, we determine the differential spectrum of $x^{p^{2l}-p^{l}+1}$. We also investigate the $c$-differential uniformity of the power mapping $x^{p^{2l}-p^{l}+1}$. In section \ref{ex code}, we calculate the value distribution of a class of exponential sum. As an application, a class of six-weight cyclic codes were presented. The conclusive remarks are given in section \ref{conclusion}.

For convenience, we fix the following notation in this paper:

\begin{tabular}{ll}
$\mathbb{F}_{p^n}$ & finite field of $p^n$ elements \\
$\mathbb{F}_{p^n}^*$ & multiplicative group of $\mathbb{F}_{p^n}$ \\
$\psi$ & a primitive element of $\mathbb{F}_{p^n}$ \\
$\Tr^n_1$ & trace function from $\mathbb{F}_{p^n}$ to $\mathbb{F}_p$ \\
$\zeta_p$ & primitive complex $p$-th root of unity \\
\end{tabular}
\section{Preliminaries}\label{Preliminaries}
 Before stating the main results, we introduce some notations. Let $e$ be a positive integer and define
\begin{equation*}
  \mu_e=\left\{x\in \mathbb{F}_{p^n}\mid x^e=1 \right\}.
\end{equation*}
If $x\in\mathbb{F}_{p^n}^*$ and $x=\psi^t$ with $0\leq t<p^n-1$, we denote $ind_{\psi}(x)=t$. Let $\mathbb{F}_{p^n}^{\sharp}=\mathbb{F}_{p^n}\setminus \{0,-1\}$. We denote
\begin{equation*}
  S_{r,s}=\{x\in\mathbb{F}_{p^n}^{\sharp}: ind_{\psi}(x+1)\equiv r\,\,({\rm mod}\,\,p^l+1), ind_{\psi}(x)\equiv s\,\,({\rm mod}\,\,p^l+1)\},
\end{equation*}
where $0\leq r,s\leq p^l$, $n=4l$. Obviously,
\begin{equation*}
 \bigsqcup\limits_{0\leq i,j\leq p^l} S_{i,j}=\mathbb{F}_{p^n}^{\sharp}
\end{equation*}
(here $\bigsqcup$ indicates disjoint union).

The following result plays a crucial role in this paper.
\begin{lemma}\label{luo wang}
([\ref{luo 111}, Theorem 1])
Let $p$ be a prime number, $n=2ks$ with $k$ and $s$ positive integers. Let $n_1$ and $n_2$ be two positive integers, $t=\gcd(n_1,n_2)$, $lcm(n_1,n_2)\mid p^s+1$. Let $0\leq r_1\leq n_1-1$ and $0\leq r_2\leq n_2-1$. Let $N_{p^n}(\chi)$ be the number of $\mathbb{F}_{p^n}$-rational points on the affine curve
\begin{equation*}
  \chi: \alpha x^{n_1}+\beta y^{n_2}+1=0.
\end{equation*}
 For $\alpha\in C_{1,r_1}=\left\{\psi^{r_1+n_1w}\mid 0\leq w\leq \frac{p^n-1}{n_1}-1 \right\}$,
 $\beta\in C_{2,r_2}=\left\{\psi^{r_2+n_2w}\mid 0\leq w\leq \frac{p^n-1}{n_2}-1 \right\}$. Then we have
\begin{itemize}
	\item[(i)] if $r_1=r_2=0$, then $N_{p^n}(\chi)=p^n+(-1)^{k-1}((n_1-1)(n_2-1)+1-t)p^{\frac{n}{2}}-t+1$.
	\item[(ii)] if $r_1=0$, $r_2\neq0$, $t\nmid r_2$, then $N_{p^n}(\chi)=p^n+(-1)^k(n_1-2)p^{\frac{n}{2}}+1$.
	\item[(iii)] if $r_1\neq0$, $r_2=0$, $t\nmid r_1$, then $N_{p^n}(\chi)=p^n+(-1)^k(n_2-2)p^{\frac{n}{2}}+1$.
	\item[(iv)] if $r_1\neq0$, $r_2\neq0$ and $t\nmid r_1-r_2$, then $N_{p^n}(\chi)=p^n+(-1)^{k-1}2p^{\frac{n}{2}}+1$.
	\item[(v)] if $r_1\neq0$, $r_2\neq0$ and $t\mid r_1-r_2$, then $N_{p^n}(\chi)=p^n+(-1)^{k}(t-2)p^{\frac{n}{2}}-t+1$.
\end{itemize}
\end{lemma}

\begin{lemma}\label{quad eq lemma}
([\ref{Tu}]). Let $n=2m$ be an even positive integer and $a,b\in\mathbb{F}_{2^n}^*$. Then the quadratic equation $x^2+ax+b=0$ has both solutions in $\mu_{2^m+1}$ if and only if
\begin{equation*}
  b=a^{1-2^m} \mbox{ and } \Tr^m_1(a^{-2^m-1})=1.
\end{equation*}
\end{lemma}

\section{Differential properties of $f(x)=x^{p^{2l}-p^{l}+1}$ over $\mathbb{F}_{p^n}$}\label{ds}

\subsection{Differential spectrum}
Let $d={p^{2l}-p^{l}+1}$, where $n=4l$, $l$ is a positive integer and $p$ is a prime number. In this subsection, we determine the differential spectrum of $f(x)=x^{d}$ over $\mathbb{F}_{p^n}$, since the differential spectrum of $x^{p^{3l}-p^{2l}+p^l}$ and $x^{p^{2l}-p^{l}+1}$  are identical. To determine the differential spectrum  of $f$, we mainly study the number of solutions in $\mathbb{F}_{p^n}$ of
\begin{equation}\label{x+1-x=b}
  (x+1)^d-x^d=b
\end{equation}
for any $b\in \mathbb{F}_{p^n}$. Denote by $\delta(b)$ the number of solutions of (\ref{x+1-x=b}) in $\mathbb{F}_{p^n}$. We have the following result.
\begin{lemma}\label{first lemma}
With the notations above, we have
\begin{itemize}
\item[(i)] $\delta(1)=p^l$.
\item[(ii)] if $b^{p^l+1}=1$ and $b\neq1$, then $\delta(b)=p^{2l}-p^l$.
\item[(iii)] if $b^{p^l+1}\neq1$, then $\delta(b)=0$ or $2$.
\end{itemize}
\end{lemma}
\begin{proof}
Let $\triangle(x)=(x+1)^{d}-x^{d}$ and $\alpha=\psi^{\frac{p^n-1}{p^l+1}}$. If $x\in\mathbb{F}_{p^n}^{\sharp}$, we can distinguish $(p^l+1)^2$ disjoint cases:
\begin{table}[H]
\centering
 \caption{Equation $(\ref{x+1-x=b})$ with $x\in\mathbb{F}_{p^n}^{\sharp}$\\}
\label{table2}
\begin{tabular}{|l|c|c|c|c|c|}
\hline
&Set & Equation $(\ref{x+1-x=b})$ & $x+1$& $x$\\
[0.5ex]
\hline
Case 1& $S_{0,0}$ & $1=b$ & To Be Determined & To Be Determined \\
\hline
Case 2& $S_{0,1}$ & $x+1-\alpha x=b$ & $\frac{b-\alpha}{1-\alpha}$ & $\frac{b-1}{1-\alpha}$\\
\hline
\vdots & \vdots  & \vdots  & \vdots  & \vdots  \\
\hline
Case $p^l+1$ & $S_{0,p^l}$ & $x+1-\alpha^{p^l} x=b$
 & $\frac{b-\alpha^{p^l}}{1-\alpha^{p^l}}$
& $\frac{b-1}{1-\alpha^{p^l}}$\\
\hline
Case $p^l+2$ & $S_{1,0}$ & $\alpha(x+1)-x=b$
 & $\frac{b-1}{\alpha-1}$
& $\frac{b-\alpha}{\alpha-1}$\\
\hline
Case $p^l+3$ & $S_{1,1}$ & $\alpha=b$
 & To Be Determined & To Be Determined \\
\hline
\vdots & \vdots  & \vdots  & \vdots  & \vdots  \\
\hline
Case $2(p^l+1)$ & $S_{1,p^l}$ & $\alpha(x+1)-\alpha^{p^l}x=b$
 & $\frac{b-\alpha^{p^l}}{\alpha-\alpha^{p^l}}$
& $\frac{b-\alpha}{\alpha-\alpha^{p^l}}$\\
\hline
\vdots & \vdots  & \vdots  & \vdots  & \vdots  \\
\hline
Case $p^l(p^l+1)+1$ & $S_{p^l,0}$ & $\alpha^{p^l}(x+1)-x=b$
 & $\frac{b-1}{\alpha^{p^l}-1}$
& $\frac{b-\alpha^{p^l}}{\alpha^{p^l}-1}$\\
\hline
Case $p^l(p^l+1)+2$ & $S_{p^l,1}$ & $\alpha^{p^l}(x+1)-\alpha x=b$
 & $\frac{b-\alpha}{\alpha^{p^l}-\alpha}$
& $\frac{b-\alpha^{p^l}}{\alpha^{p^l}-\alpha}$\\
\hline
\vdots & \vdots  & \vdots  & \vdots  & \vdots  \\
\hline
Case $(p^l+1)^2$ & $S_{p^l,p^l}$ & $\alpha^{p^l}=b$
 &  To Be Determined & To Be Determined \\
\hline
\end{tabular}
\end{table}

(i) Note that $\triangle(x)=1$ has no solution in $\bigsqcup\limits_{1\leq i\leq p^l} S_{i,i}$. If $\triangle(x)=1$ has one solution $x_1\in S_{j_1,j_2}$, where $j_1\neq j_2$ and  $j_1,j_2\neq0$, the
\begin{equation*}
  x_1=\frac{1-\alpha^{j_1}}{\alpha^{j_1}-\alpha^{j_2}} \mbox{ and }  x_1+1=\frac{1-\alpha^{j_2}}{\alpha^{j_1}-\alpha^{j_2}}.
\end{equation*}
Notice that $\alpha^i\in \mu_{p^l+1}$ for $i=1,2,\cdots,p^l$ which implies $\alpha^{ip^l}=\alpha^{ip^{3l}}=\alpha^{-i}$ and $\alpha^{ip^{2l}}=\alpha^i$. Then we have
\begin{equation*}
  x_1^{p^{3l}-p^{2l}+p^l-1}=\left(\frac{1-\alpha^{j_1}}{\alpha^{j_1}-\alpha^{j_2}}\right)^{p^{3l}-p^{2l}+p^l-1}
  =\alpha^{2j_2}=\alpha^{j_2},
\end{equation*}
\begin{equation*}
  (x_1+1)^{p^{3l}-p^{2l}+p^l-1}=\left(\frac{1-\alpha^{j_2}}{\alpha^{j_1}-\alpha^{j_2}}\right)^{p^{3l}-p^{2l}+p^l-1}
  =\alpha^{2j_1}=\alpha^{j_1}.
\end{equation*}
But this contradicts ${j_1\neq j_2}$. So $\triangle(x)=1$ has no solution in $\mathbb{F}_{p^n}^{\sharp}\setminus S_{0,0}$. If $x\in S_{0,0}$, then $\triangle(x)=1$ is equivalent to the system of equations
\begin{equation}\label{001b}
    \left \{
\begin{array}{lll}
y-x=1,\\
ind_{\psi}(y)=ind_{\psi}(x)\equiv 0\,\,({\rm mod}\,\,p^l+1).
\end{array}\right.
\end{equation}
By Lemma \ref{luo wang}, we have $p^l-2$ solutions to (\ref{001b}). We also get $\triangle(0)=1$ and $\triangle(-1)=1$. In total, $\delta(1)=p^l$.

(ii) We only show the case $b=\alpha$. The other cases for $b=\alpha^k$, $2\leq k\leq p^l$ are similar and we omit them. Note that $\triangle(x)=\alpha$ has no solution in $\bigsqcup\limits_{\substack {0\leq i\leq p^l,\\i\neq1}} S_{i,i}$. If $\triangle(x)=\alpha$ has one solution $y_1\in S_{k_1,k_2}$ with $k_1\neq k_2$ and $k_1,k_2\neq1$, then
\begin{equation*}
  y_1=\frac{\alpha-\alpha^{k_1}}{\alpha^{k_1}-\alpha^{k_2}} \mbox{ and } y_1+1=\frac{\alpha-\alpha^{k_2}}{\alpha^{k_1}-\alpha^{k_2}}.
\end{equation*}
Then we obtain
\begin{equation*}
  y_1^{p^{3l}-p^{2l}+p^l-1}=\left(\frac{\alpha-\alpha^{k_1}}{\alpha^{k_1}-\alpha^{k_2}}\right)^{p^{3l}-p^{2l}+p^l-1}
  =\left(\frac{\alpha^{p^l}\alpha^{k_1}-1}{\alpha-\alpha^{k_1}}\right)^2\alpha^{2k_2}=\alpha^{k_2},
\end{equation*}
\begin{equation*}
  (y_1+1)^{p^{3l}-p^{2l}+p^l-1}=\left(\frac{\alpha-\alpha^{k_2}}{\alpha^{k_1}-\alpha^{k_2}}\right)^{p^{3l}-p^{2l}+p^l-1}
  =\left(\frac{\alpha^{p^l}\alpha^{k_2}-1}{\alpha-\alpha^{k_2}}\right)^2\alpha^{2k_1}=\alpha^{k_1},
\end{equation*}
which yields
\begin{equation}\label{fuza6}
  \alpha^{k_1}+\alpha^{k_2}=\alpha^{-2}\alpha^{k_1}\alpha^{k_2}+2\alpha-1.
\end{equation}
Taking $p^{l+1}$-th power on both sides of (\ref{fuza6}), it implies
\begin{equation}\label{fuza7}
\alpha^{k_1}\alpha^{k_2}=\alpha^{2},
\end{equation}
which yields
\begin{equation*}
    \left \{
\begin{array}{lll}
\alpha^{k_1}\alpha^{k_2}=\alpha^2,\\
\alpha^{k_1}+\alpha^{k_2}=2\alpha.
\end{array}\right.
\end{equation*}
Hence, $\alpha^{k_1}=\alpha^{k_2}=\alpha$, which contradicts ${k_1\neq k_2}$. So $\triangle(x)=\alpha$ has no solution in $\mathbb{F}_{p^n}^{\sharp}\setminus S_{1,1}$. If $x\in S_{1,1}$, then $\triangle(x)=\alpha$ is equivalent to the system of equations
\begin{equation}\label{001bpsi}
    \left \{
\begin{array}{lll}
\psi y-\psi x=1,\\
ind_{\psi}(y)=ind_{\psi}(x)\equiv 0\,\,({\rm mod}\,\,p^l+1).\\
\end{array}\right.
\end{equation}
By Lemma \ref{luo wang}, we have $p^{2l}-p^l$ solutions to (\ref{001bpsi}). We also get $\triangle(0)\neq\alpha$ and $\triangle(-1)\neq\alpha$. In total, $\delta(\alpha)=p^{2l}-p^l$.

(iii) Note that
\begin{equation}\label{dpl+1}
  d(p^l+1)\equiv p^l+1\,\,({\rm mod}\,\,p^{4l}-1).
\end{equation}
Using (\ref{x+1-x=b}) and (\ref{dpl+1}), one obtains
\begin{equation}\label{imp}
  (x+1)^{p^l+1}=(x^d+b)^{p^l+1}.
\end{equation}
According to (\ref{imp}), we have
\begin{equation}\label{r1u1}
  x^d+b=\gamma_1(x+1),
\end{equation}
for some $\gamma_1\in \mu_{p^l+1}$.
Taking $(p^l+1)$-th power on both sides of (\ref{r1u1}), we have
\begin{equation*}
  x^{p^l+1}=(\gamma_1(x+1)-b)^{p^l+1}.
\end{equation*}
Then we obtain
\begin{equation}\label{r2u2}
  \gamma_2x=\gamma_1(x+1)-b,
\end{equation}
for some $\gamma_2\in \mu_{p^l+1}$.
When $\gamma_1=\gamma_2$, we obtain $b=\gamma_1\in \mu_{p^l+1}$. This situation has been discussed in (i).
When $\gamma_1\neq\gamma_2$, we have $x=\frac{b-\gamma_1}{\gamma_1-\gamma_2}$ which is a solution of (\ref{x+1-x=b}) if and only if
\begin{equation}\label{fuzaxd}
  \left(\frac{b-\gamma_2}{\gamma_1-\gamma_2}\right)^d-
  \left(\frac{b-\gamma_1}{\gamma_1-\gamma_2}\right)^d=b.
\end{equation}
Notice that $\gamma_i\in \mu_{p^l+1}$ for $i=1,2$ which implies $\gamma_i^{p^l}=\gamma_i^{p^{3l}}=\gamma_i^{-1}$ and $\gamma_i^{p^{2l}}=\gamma_i$. Then by a detailed calculation one can obtain that
\begin{equation*}
  (b-\gamma_i)^d
=\frac{b^{p^{3l}+p^l}-(b^{p^{3l}}+b^{p^l})\gamma_i^{p^l}+\gamma_i^{2p^l}}{b^{p^{2l}}-\gamma_i}
\end{equation*}
for any $b^{p^{2l}}\neq\gamma_i$ and from which one can also have
\begin{equation*}
  (\gamma_1-\gamma_2)^d=\frac{\gamma_1-\gamma_2}{\gamma_1^2\gamma_2^2}.
\end{equation*}
Then (\ref{fuzaxd}) can be rewritten as
\begin{equation}\label{fuza2}
\begin{aligned}
  &(\gamma_2-\gamma_1)b^{p^{3l}+p^l}+(b^{p^{3l}+p^{2l}}
  +b^{p^{2l}+p^l}+b)\left(\frac{1}{\gamma_1}-\frac{1}{\gamma_2}\right)
  +(b^{p^{3l}}+b^{p^l})\left(\frac{\gamma_1}{\gamma_2}-\frac{\gamma_2}{\gamma_1}\right)\\
  &+(b^{p^{2l}}+b^{p^{2l}+1})\left(\frac{1}{\gamma_2^2}-\frac{1}{\gamma_1^2}\right)
  +\frac{\gamma_2}{\gamma_1^2}-\frac{\gamma_1}{\gamma_2^2}
  -b^{2p^{2l}+1}\left(\frac{1}{\gamma_1\gamma_2^2}-\frac{1}{\gamma_2\gamma_1^2}\right)=0.
\end{aligned}
\end{equation}
Taking $p^{2l}$-th power on both sides of (\ref{fuza2}) and subtracting these two equations lead to
\begin{equation}\label{fuza3}
(b^{p^{2l}}-b)\left[(b^{p^{3l}}+b^{p^{l}}-1)\left(\frac{1}{\gamma_1}-\frac{1}{\gamma_2}\right)
+\left(\frac{1}{\gamma_2^2}-\frac{1}{\gamma_1^2}\right)
+b^{p^{2l}+1}\left(\frac{1}{\gamma_1\gamma_2^2}-\frac{1}{\gamma_2\gamma_1^2}\right)\right]=0.
\end{equation}

\emph{Case 1:} $b^{p^{2l}}\neq b$. Then by ($\ref{fuza3}$) one gets
\begin{equation}\label{fuza4}
  \gamma_1+\gamma_2=(b^{p^{3l}}+b^{p^{l}}-1)(\gamma_1\gamma_2)-b^{p^{2l}+1}.
\end{equation}
Taking $p^{l+1}$-th power on both sides of (\ref{fuza4}), it implies
\begin{equation}\label{fuza5}
(b^{p^{3l}}-1)(1-b^{p^l})(\gamma_1\gamma_2)^2
+(b^{p^{2l}}-1)(b-1)\gamma_1\gamma_2=0.
\end{equation}
Then we have
\begin{equation*}
    \left \{
\begin{array}{lll}
\gamma_1\gamma_2=(b^{p^{3l}}-1)(1-b^{p^l})^{p^l-1},\\
\gamma_1+\gamma_2
=(b^{p^{3l}}+b^{p^{l}}-1)(b^{p^{3l}}-1)(1-b^{p^l})^{p^l-1}-b^{p^{2l}+1},\\
\end{array}\right.
\end{equation*}
i.e., $\gamma_1,\gamma_2$ are two solutions of
\begin{equation*}
  x^2-\big((b^{p^{3l}}+b^{p^{l}}-1)(b^{p^{3l}}-1)(1-b^{p^l})^{p^l-1}-b^{p^{2l}+1}\big)
  x+(b^{p^{3l}}-1)(1-b^{p^l})^{p^l-1}=0.
\end{equation*}
Hence, when $\gamma_1\neq\gamma_2$, $b\not\in \mu_{p^l+1}$ and $b^{p^{2l}}\neq b$, the equation (\ref{x+1-x=b}) has either no solution or two solutions.

\emph{Case 2:} $b^{p^{2l}}= b$. For a fixed $b$, if $\triangle(x_2)=b$, where
\begin{equation*}
  x_2=\frac{b-\alpha^{j_3}}{\alpha^{j_3}-\alpha^{j_4}}\in S_{j_3,j_4},
\end{equation*}
then we obtain
\begin{equation*}
  x_2^{p^{3l}-p^{2l}+p^l-1}=\left(\frac{b-\alpha^{j_3}}{\alpha^{j_3}-\alpha^{j_4}}\right)^{p^{3l}-p^{2l}+p^l-1}
  =\left(\frac{b^{p^l}\alpha^{j_3}-1}{b-\alpha^{j_3}}\right)^2\alpha^{2j_4}=\alpha^{j_4},
\end{equation*}
\begin{equation*}
  (x_2+1)^{p^{3l}-p^{2l}+p^l-1}=\left(\frac{b-\alpha^{j_4}}{\alpha^{j_3}-\alpha^{j_4}}\right)^{p^{3l}-p^{2l}+p^l-1}
  =\left(\frac{b^{p^l}\alpha^{j_4}-1}{b-\alpha^{j_4}}\right)^2\alpha^{2j_3}=\alpha^{j_3},
\end{equation*}
which yields
\begin{equation}\label{fuza61}
  \alpha^{j_3}+\alpha^{j_4}=b^{2p^l}\alpha^{j_3}\alpha^{j_4}+2b-1.
\end{equation}
Taking $p^{l+1}$-th power on both sides of (\ref{fuza61}), it implies
\begin{equation}\label{fuza71}
\alpha^{j_3}\alpha^{j_4}=(b^{p^l}-1)^{2(p^l-1)}.
\end{equation}
Then by (\ref{fuza61}) and (\ref{fuza71}) one can obtain
\begin{equation*}
    \left \{
\begin{array}{lll}
\alpha^{j_3}\alpha^{j_4}=(b^{p^l}-1)^{2(p^l-1)},\\
\alpha^{j_3}+\alpha^{j_4}=b^{2p^l}(b^{p^l}-1)^{2(p^l-1)}+2b-1,\\
\end{array}\right.
\end{equation*}
i.e., $\alpha^{j_3}, \alpha^{j_4}$ are two solutions of
\begin{equation*}
  x^2-(b^{2p^l}(b^{p^l}-1)^{2(p^l-1)}+2b-1)
  x+(b^{p^l}-1)^{2(p^l-1)}=0.
\end{equation*}
Hence, when $\gamma_1\neq\gamma_2$, $b\not\in \mu_{p^l+1}$ and $b^{p^{2l}}= b$, the equation (\ref{x+1-x=b}) has either no solution or two solutions. This completes the proof.
\end{proof}

We determine the differential spectrum of $f$ in the following result.

\begin{theorem}
Let $f(x)=x^d$ be a power mapping defined over $\mathbb{F}_{p^n}$, where $d=p^{3l}-p^{2l}+p^l$, $n=4l$. Then the differential spectrum of $f$ is given by
\begin{equation*}
  \mathbb{DS}=\left\{\omega_0=\frac{p^n+p^{3l}-p^{2l}-p^l-2}{2}, \omega_2=\frac{p^n-p^{3l}+p^{2l}-p^l}{2}, \omega_{p^l}=1,
  \omega_{p^{2l}-p^l}=p^l\right \}.
\end{equation*}
\end{theorem}

\begin{proof}
From (\ref{two basic ide}) and Lemma \ref{first lemma}, we obtain
\begin{equation}\label{1 eq1}
  \omega_0+\omega_2+\omega_{e_1}+
  \omega_{e_2}=p^n
\end{equation}
\begin{equation}\label{1 eq2}
  2\cdot\omega_2+p^l\cdot\omega_{p^l}+
  (p^{2l}-p^l)\cdot\omega_{p^{2l}-p^l}=p^n
\end{equation}
\begin{equation}\label{1 eq3}
  \omega_{p^l}=1
\end{equation}
\begin{equation}\label{1 eq4}
  \omega_{p^{2l}-p^l}=p^l.
\end{equation}
Solving the system equations consisting of (\ref{1 eq1})-(\ref{1 eq4}) we obtain the result.
\end{proof}

\subsection{$c$-differential uniformity of $x^{p^{2l}-p^{l}+1}$ over $\mathbb{F}_{p^n}$}
We give the following result on the $c$-differential uniformity of $f(x)=x^{p^{3l}-p^{2l}+p^l}$ for $c\in\mathbb{F}_{p^n}\setminus \mu_{p^l+1}$, where $n=4l$. For $c\in\mathbb{F}_{p^n}\setminus \mu_{p^l+1}$, the $c$-differential uniformity of $x^{p^{2l}-p^{l}+1}$ and $x^{p^{3l}-p^{2l}+p^l}$ are the same.

\begin{theorem}\label{xiaoyu3}
Let $d=p^{3l}-p^{2l}+p^l$. For $c\in\mathbb{F}_{p^n}\setminus \mu_{p^l+1}$, we have $_c\Delta_f\leq(p^l+1)^2$.
\end{theorem}
\begin{proof}
For a fixed $c\in\mathbb{F}_{p^n}\setminus \mu_{p^l+1}$, we mainly study the number of solutions in $\mathbb{F}_{p^n}$ of
\begin{equation}\label{x+1-cx=b}
   \triangle(x)=(x+1)^{d}-cx^{d}=b
\end{equation}
for any $b\in \mathbb{F}_{p^n}$. Denote by $\delta_{c}(b)$ the number of solutions of (\ref{x+1-cx=b}) in $\mathbb{F}_{p^n}$. Let $\alpha=\psi^{(p^n-1)/(p^l+1)}$. If $x\in\mathbb{F}_{p^n}^{\sharp}$, we can distinguish $(p^l+1)^2$ disjoint cases:
\begin{table}[H]
\centering
 \caption{Equation $(\ref{x+1-cx=b})$ with $x\in\mathbb{F}_{p^n}^{\sharp}$\\}
\label{table4}
\begin{tabular}{|l|c|c|c|c|c|}
\hline
&Set & Equation $(\ref{x+1-cx=b})$ & $x+1$& $x$\\
[0.5ex]
\hline
Case 1& $S_{0,0}$ & $x+1-cx=b$ & $\frac{b-c}{1-c}$ & $\frac{b-1}{1-c}$ \\
\hline
Case 2& $S_{0,1}$ & $x+1-\alpha cx=b$ & $\frac{b-\alpha c}{1-\alpha c}$ & $\frac{b-1}{1-\alpha c}$\\
\hline
\vdots & \vdots  & \vdots  & \vdots  & \vdots  \\
\hline
Case $p^l+1$ & $S_{0,p^l}$ & $x+1-\alpha^{p^l} cx=b$
 & $\frac{b-\alpha^{p^l}c}{1-\alpha^{p^l}c}$
& $\frac{b-1}{1-\alpha^{p^l}c}$\\
\hline
Case $p^l+2$ & $S_{1,0}$ & $\alpha(x+1)-cx=b$
 & $\frac{b-c}{\alpha-c}$
& $\frac{b-\alpha}{\alpha-c}$\\
\hline
\vdots & \vdots  & \vdots  & \vdots  & \vdots  \\
\hline
Case $2p^l+2$ & $S_{1,p^l}$ & $\alpha(x+1)-\alpha^{p^l}cx=b$
 & $\frac{b-\alpha^{p^l}c}{\alpha-\alpha^{p^l}c}$
& $\frac{b-\alpha}{\alpha-\alpha^{p^l}c}$\\
\hline
\vdots & \vdots  & \vdots  & \vdots  & \vdots  \\
\hline
Case $p^{2l}+p^l+1$ & $S_{p^l,0}$ & $\alpha^{p^l}(x+1)-cx=b$
 & $\frac{b-c}{\alpha^{p^l}-c}$
& $\frac{b-\alpha^{p^l}}{\alpha^{p^l}-c}$\\
\hline
\vdots & \vdots  & \vdots  & \vdots  & \vdots  \\
\hline
Case $(p^l+1)^2$ & $S_{p^l,p^l}$ & $\alpha^{p^l}(x+1)-\alpha^{p^l}cx=b$
 & $\frac{b-\alpha^{p^l}c}{\alpha^{p^l}-\alpha^{p^l}c}$ & $\frac{b-\alpha^{p^l}}{\alpha^{p^l}-\alpha^{p^l}c}$ \\
\hline
\end{tabular}
\end{table}
From Table \ref{table4}, we have $\delta_c(b)\leq(p^l+1)^2$ for $b\neq1,c$ since $\triangle(0)=1$ and $\triangle(-1)=c$. Note that $\triangle(x)=1$ has no solution in $\bigsqcup\limits_{0\leq i\leq p^l} S_{0,i}$, and $\triangle(x)=c$ has no solution in $\bigsqcup\limits_{0\leq i\leq p^l} S_{i,0}$ and so, $\delta_{c}(1),\delta_{c}(c)\leq(p^l+1)^2$. Then, according to $\gcd(d,p^n-1)\leq3$ and lemma \ref{c delta f}, we get $_c\Delta_f\leq(p^l+1)^2$.
\end{proof}

\section{On Exponential sum}\label{ex code}

In [\ref{H1}, \ref{H2}], Helleseth determined the value distribution of the cross correlation function for $d=p^{3l}-p^{2l}+p^l$ with $n=4l$, $p^{l}\not\equiv 2\,\,({\rm mod}\,\,3)$. So what happens if $p^{l}\equiv 2\,\,({\rm mod}\,\,3)$? In this section, we give some relevant conclusions. We fix $n=4l$, where $l$ is a positive integer. Let
$d_1=p^{3l}-p^{2l}+p^l$, where $p^{l}\equiv 2\,\,({\rm mod}\,\,3)$.
Define
\begin{equation}\label{suv}
S_{d_1}(u,v)=\sum\limits_{x\in\mathbb{F}_{p^n}}\zeta_p^{{\rm{Tr}}_1^n(ux^{d_1}-vx)},
\end{equation}
where $(u,v)\in \mathbb{F}_{p^n}\times\mathbb{F}_{p^n}$, $\zeta_p=e^{\frac{2\pi\sqrt{-1}}{p}}$ is a primitive complex $p$-th root of unity, and $\Tr^n_1(x)=\sum\limits_{i=0}^{n-1}x^{p^i}$ denotes the absolute trace mapping from $\mathbb{F}_{p^n}$ to $\mathbb{F}_{p}$.  Define
\begin{equation*}
\begin{aligned}
&C_0=\left\{x\in \mathbb{F}_{p^n}^{*}\mid x=y^{p^l+1}\mbox{ for some }y \in\mathbb{F}_{p^n}\right\},\\
&C_{\infty}=\left\{0\right\},\\
&C_1=\mathbb{F}_{p^n}\setminus(C_0\cup C_{\infty}).
\end{aligned}
\end{equation*}

In order to determine the values of $S_{d_1}(u,v)$. We need one result without proof for which the readers can refer to [\ref{1972 ex}].
\begin{lemma}\label{jielun}
Suppose $n=4l$ and $\zeta_p$ be a complex $p$-th root of unity. Then
\begin{equation*}
    \sum\limits_{y\in \mathbb{F}_{p^n}}\zeta_p^{\Tr ^n_1(ay^{p^{l}+1})}=
	\left\{ \begin{array}{llll}
		p^{n}&\mbox{if}&a\in C_{\infty},\\
		-p^{3l}&\mbox{if}&a\in C_0,\\
		p^{2l}&\mbox{if}&a\in C_1.
	\end{array}\right.
\end{equation*}
\end{lemma}

The following result is essential to determine the value of $S_{d_1}(u,v)$.
\begin{lemma}\label{1 alpha}
Let $d_1=p^{3l}-p^{2l}+p^l$,
where $p^{l}\equiv 2\,\,({\rm mod}\,\,3)$. For each fixed $u \in \mathbb{F}_{p^n}^*$, we have
\begin{equation*}
    \{S_{d_1}(u,v): v\in\mathbb{F}_{p^{n}}^* \}=\left\{ \begin{array}{llll}
	\{S_{d_1}(1,v): v\in\mathbb{F}_{p^{n}}^* \},&\mbox{if}&u \mbox{ is a cube in }\mathbb{F}_{p^{n}}^*,\\
	\{S_{d_1}(\psi,v): v\in\mathbb{F}_{p^{n}}^* \},&\mbox{if}&u \mbox{ is a non-cube in }\mathbb{F}_{p^{n}}^*.
	\end{array}\right.
\end{equation*}
\end{lemma}
\begin{proof}
  Since $\gcd(d_1,p^n-1)=3$, there exist two integers $u_1$ and $v_1$ such that $u_1d_1+v_1(p^n-1)=3$. Then for any $x\in \mathbb{F}_{p^n}^*$, we have $x^3=x^{u_1d_1}$. Moreover, note that every $u \in \mathbb{F}_{p^n}^*$ can be written uniquely as $u={\psi}^{3k_1+j}$ for some $0\leq k_1< (p^n-1)/3$ and $0\leq j\leq2$. Replacing $x$ by $\psi^{-k_1u_1}x$ in $S_{d_1}(u,v)$ gives
\begin{equation*}
    S_{d_1}(u,v)=\left\{ \begin{array}{llll}
	S_{d_1}(1,v\psi^{-k_1u_1}),&\mbox{if}&j=0 ,\\
	S_{d_1}(\psi,v\psi^{-k_1u_1}),&\mbox{if}&j=1 .
	\end{array}\right.
\end{equation*}
If $j=2$, denote $p=3k_2+2$, $y=\psi^{(k_1-k_2)u_1}x$ and $z=y^{p^{n-1}}$, then
\begin{equation*}
  \begin{aligned}
  S_{d_1}(u,v)&=\sum\limits_{x\in\mathbb{F}_{p^n}}\zeta_p^{{\rm{Tr}}_1^n(ux^{d_1}-vx)}\\
  &=\sum\limits_{y\in \mathbb{F}_{p^n}}\zeta_p^{\Tr_1^{n}(\psi^py^{d_1}-v\psi^{(k_2-k_1)u_1}y)}\\
  &=\sum\limits_{z\in \mathbb{F}_{p^n}}\zeta_p^{\Tr_1^{n}(\psi z^{d_1}-v^{p^{n-1}}\psi^{p^{n-1}(k_2-k_1)u_1}z)}\\
  &=S_{d_1}(\psi,v^{p^{n-1}}\psi^{p^{n-1}(k_2-k_1)u_1}).\\
\end{aligned}
\end{equation*}
Because $v\psi^{-ku_1}$ and $v^{p^{n-1}}\psi^{p^{n-1}(k_2-k_1)u_1}$ all run through $\mathbb{F}_{p^n}^*$ as $v$ runs through $\mathbb{F}_{p^n}^*$, we can obtain the result.
\end{proof}

We define
\begin{equation*}
    n_k(u,v)=\# \left\{j\mid u\psi^{d_1j}-v\psi^j\in C_k, 0\leq j\leq p^{l} \right\}
\end{equation*}
for $k=0,1,\infty$ and $(u,v)\in \mathbb{F}_{p^n}^*\times\mathbb{F}_{p^n}$.
\begin{lemma}\label{n_i(u,v)}
Let $n=4l$, where $l$ is a positive integer. Let $d_1=p^{3l}-p^{2l}+p^{l}$, where $p^{l}\equiv 2\,\,({\rm mod}\,\,3)$. Let $(u,v)\in \mathbb{F}_{p^n}^*\times\mathbb{F}_{p^n}$, then the following holds
\begin{itemize}
	\item[(i)] $n_{\infty}(u,v)=1$ if $({\frac{v}{u}})^{p^l+1}=1$ and $n_{\infty}(u,v)=0$ otherwise.
	\item[(ii)] $n_0(u,v)+n_{\infty}(u,v)\leq3$.
 \item[(iii)] If $u=1$ and $n_{\infty}(u,v)=1$, then $n_0(u,v)=0$ for $v=1,\psi^{(d_1-1)\frac{p^l+1}{3}},\psi^{2(d_1-1)\frac{p^l+1}{3}}$ and $n_0(u,v)=1$ otherwise.
 \item[(iv)] If $u=\psi$ and $n_{\infty}(u,v)=1$, then $n_0(u,v)=1$.
\end{itemize}
\end{lemma}
\begin{proof}
 (i) If $u\psi^{d_1j}-v\psi^j=0$, then $\psi^{(d_1-1)j}=\frac{v}{u}$. Note that
\begin{equation*}
    d_1-1=p^{3l}-p^{2l}+p^{l}-1=\frac{p^{n}-1}{p^{l}+1}.
\end{equation*}
It follows that $n_{\infty}(u,v)=1$ if $({\frac{v}{u}})^{p^l+1}=1$ and $n_{\infty}(u,v)=0$ otherwise.

(ii) Let $T(u,v)=n_0(u,v)+n_{\infty}(u,v)$. Since there are $n_{\infty}(u,v)$ solutions corresponding to $ux^{d_1}-vx=0$ and $n_0(u,v)$ solutions corresponding to $ux-vx^{d_1}\in C_0$. Then $T(u,v)$ is the number of solutions of \begin{equation*}
    (ux^{d_1}-vx)^{p^{3l}+p^{l}}-(ux^{d_1}-vx)^{p^{2l}+1}=0
\end{equation*}
where $x=\psi^j$, $0\leq j\leq p^l$. Straightforward calculation implies that $T(u,v)$ is the number of solutions of
\begin{equation*}
\begin{aligned}
    u^{p^{2l}+1}x^{3(d_1-1)}-(v^{p^{3l}+p^l}+v^{p^{2l}}u+vu^{p^{2l}})x^{2(d_1-1)}
    +(v^{p^{2l}+1}+v^{p^{3l}}u^{p^l}+v^{p^l}u^{p^{3l}})x^{d_1-1}-u^{p^{3l}+p^l}=0,
\end{aligned}
\end{equation*}
where $x=\psi^j$, $0\leq j\leq p^l$. Therefore, by letting $y=x^{d_1-1}$, we obtain that $T(u,v)$ is the number of solutions to
\begin{equation*}
     \left \{
     \begin{array}{lll}
u^{p^{2l}+1}y^3-(v^{p^{3l}+p^l}+v^{p^{2l}}u+vu^{p^{2l}})y^2
    +(v^{p^{2l}+1}+v^{p^{3l}}u^{p^l}+v^{p^l}u^{p^{3l}})y-u^{p^{3l}+p^l}=0.\\
    y^{p^l+1}=1.\\
    \end{array}\right.
\end{equation*}
It follows that $n_0(u,v)+n_{\infty}(u,v)\leq3$.

(iii) In the case $n_{\infty}(u,v)=1$, it follows from (1) that $({\frac{v}{u}})^{p^l+1}=1$. This implies that $u^{p^l+1}=v^{p^l+1}$.

\emph{Case 1:} $u=v=1$. Then $T(u,v)$ is the number  of common solutions of $y^{p^l+1}=1$ and
\begin{equation*}
    y^3-3y^2+3y-1={(y-1)}^3=0.
\end{equation*}
So, we obtain $T(u,v)=n_0(u,v)+n_{\infty}(u,v)=1$.

\emph{Case 2:} $v=\psi^{(d_1-1)j},\,u=\psi^{(d_1-1)j}$, where $d_1-1=\frac{p^{4l}-1}{p^l+1}$ for $1\leq j\leq p^l$. Then $T(u,v)$ is the number  of common solutions of $y^{p^l+1}=1$ and
\begin{equation*}
    \begin{aligned}
        &y^3-(\psi^{-2j(d_1-1)}+2\psi^{j(d_1-1)})y^2+(\psi^{2j(d_1-1)}+2\psi^{-j(d_1-1)})y-1\\
     &=(y-\psi^{-2j(d_1-1)}){(y-\psi^{j(d_1-1)})}^2\\
     &=0.
    \end{aligned}
\end{equation*}
Then $n_0(u,v)=0$ for $j=\frac{p^l+1}{3},\frac{2(p^l+1)}{3}$ and $n_0(u,v)=1$ otherwise.

(iv) In the case $n_{\infty}(u,v)=1$, it follows from (1) that $({\frac{v}{u}})^{p^l+1}=1$. This implies that $u^{p^l+1}=v^{p^l+1}$.

\emph{Case 1:} $v=u=\psi$. Then $T(a,b)$ is the number  of common solutions of $y^{p^l+1}=1$ and
\begin{equation*}
\begin{aligned}
    &\psi^{p^{2l}+1}y^3-(\psi^{p^{3l}+p^l}+2\psi^{p^{2l}+1})y^2
    +(\psi^{p^{2l}+1}+2\psi^{p^{3l}+p^l})y-\psi^{p^{3l}+p^l}\\
    &=(\psi^{p^{2l}+1}y-\psi^{p^{3l}+p^l}){(y-1)}^2\\
    &=0.
    \end{aligned}
\end{equation*}
So, we have $T(u,v)=n_0(u,v)+n_{\infty}(u,v)=2$.

\emph{Case 2:} $v=\psi^{(d_1-1)j}u=\psi^{(d_1-1)j}\psi$, where $d_1-1=\frac{p^{4l}-1}{p^l+1}$, $j=1,2,...,p^l$. Then $T(u,v)$ is the number  of common solutions of $y^{p^l+1}=1$ and
\begin{equation*}
 \begin{aligned}
    &\psi^{p^{2l}+1}y^3-(\psi^{-2j(d_1-1)}\psi^{p^{3l}+p^l}
    +2\psi^{j(d_1-1)}\psi^{p^{2l}+1})y^2\\
    &+(\psi^{2j(d_1-1)}\psi^{p^{2l}+1}
    +2\psi^{-j(d_1-1)}\psi^{p^{3l}+p^l})y-\psi^{p^{3l}+p^l}\\
    &=(\psi^{p^{2l}+1}y-\psi^{-2j(d_1-1)}\psi^{p^{3l}+p^l}){(y-\psi^{j(d_1-1)})}^2\\
    &=0.
     \end{aligned}
\end{equation*}
So, we get $T(u,v)=n_0(u,v)+n_{\infty}(u,v)=2$.
\end{proof}

Sum it up, $S_{d_1}(u,v)$ can be explicitly determined.
\begin{theorem}\label{wf2 value}
Let $p$ be a prime number and $p^l\equiv 2\,\,({\rm mod}\,\,3)$. Let $d_1=p^{3l}-p^{2l}+p^l$. When $(u,v)$ runs through $\mathbb{F}_{p^n}^*\times\mathbb{F}_{p^n}$, $S_{d_1}(u,v)$ takes values in
\begin{equation*}
\left\{-2p^{2l}, -p^{2l}, 0, p^{2l}, p^{3l}-p^{2l}, p^{3l}\right\}.
\end{equation*}
\end{theorem}
\begin{proof}
     Note that every $x\in \mathbb{F}_{p^n}^*$ can be written as $x=\psi^{(p^l+1)k+j}$, where $0\leq k<\frac{p^n-1}{p^l+1}=K$ and $0\leq j\leq p^l$. From $\psi^{(p^l+1)d}=\psi^{p^l+1}$ and Lemma \ref{jielun} we have
\begin{equation*}
\begin{aligned}
     S_{d_1}(u,v)&=\sum\limits_{x\in \mathbb{F}_{p^n}}\zeta_p^{\Tr ^n_1(ux^{d_1}-vx)}\\
&=\displaystyle\sum_{j=0}^{p^l}\displaystyle\sum_{k=0}^{K}\zeta_p^{\Tr ^n_1(u\psi^{(p^l+1)d_1k+dj}-v{\psi^{(p^l+1)k+j})}}+1\\
&=\displaystyle\sum_{j=0}^{p^l}\displaystyle\sum_{k=0}^{K}\zeta_p^{\Tr ^n_1(\psi^{(p^l+1)k}(u{\psi^{d_1j}}-v\psi^{j}))}+1\\
&=\frac{1}{p^l+1} \big(p^{4l}n_\infty(u,v)-p^{3l}n_0(u,v)+p^{2l}n_1(u,v)\big).
\end{aligned}
\end{equation*}

It is therefore sufficient to find the distribution of $(n_\infty(u,v),n_0(u,v),n_1(u,v))$ when $(u,v)\in\mathbb{F}_{p^n}^*\times\mathbb{F}_{p^n}$.
According to Lemma \ref{n_i(u,v)}, when $u=1$ or $\psi$ ,the triplet $(n_\infty(u,v),n_0(u,v),n_1(u,v))$ and the value of $S_{d_1}(u,v)$ have the following possibilities:
\[
\begin{array}{cccc}
n_{\infty}(u,v) &  n_0(u,v) &  n_1(u,v) & S_{d_1}(u,v) \\
0  &  0  & p^l+1 & -1+p^{2l} \\
0  &  1  & p^l & -1 \\
0  &  2  & p^l-1 & -1-p^{2l} \\
0  &  3  & p^l-2 & -1-2p^{2l} \\
1  &  0  & p^l & -1+p^{3l} \\
1  &  1  & p^l-1 & -1+p^{3l}-p^{2l}
\end{array}
\]
By Lemma \ref{1 alpha}, the desired result follows.
\end{proof}

To determine the multiplicity on each value of $S_{d_1}(u,v)$ for $(u,v)\in\mathbb{F}_{p^n}^*\times\mathbb{F}_{p^n}$. We also need several moment identities.
\begin{lemma}\label{wf2 3powersum}
For the exponential sum $S_{d_1}(u,v)$, the following identities hold:
\begin{itemize}
\item[(i)] $\sum\limits_{u,v\in \mathbb{F}_{p^n}}S_{d_1}(u,v)=p^{2n}$.
\item[(ii)]$\sum\limits_{u,v\in \mathbb{F}_{p^n}}S_{d_1}(u,v)^2=p^{3n}$.
\item[(iii)] $\sum\limits_{u,v\in \mathbb{F}_{p^n}}S_{d_1}(u,v)^3=p^{13l}+p^{3n}-p^{9l}$.
\end{itemize}
\end{lemma}
\begin{proof}
The proof of (i) and (ii) is obvious and we omit it here. For (iii),  we can calculate
\begin{equation*}
    \sum\limits_{u,v\in \mathbb{F}_{p^n}}S_{d_1}(u,v)^3=p^{2n}\cdotp \big(3p^n-2+(p^n-1)\cdotp M_1\big)
\end{equation*}
where $M_1$ is the number of solutions to
\begin{equation*}
    \left \{
\begin{array}{lll}y_1+y_2+1=0,\\
y_1^{d_1}+y_2^{d_1}+1=0,\\
\end{array}\right.
\end{equation*}
satisfying $y_1y_2\neq0$. Since $d_1=p^{3l}-p^{2l}+p^l$, this equals the number of solutions of \begin{equation*}
    y^{p^{3l}+p^l}(-y-1)^{p^{2l}}+(y+1)^{p^{3l}+p^l}y^{p^{2l}}-y^{p^{2l}}(y+1)^{p^{2l}}=0
\end{equation*} in $\mathbb{F}_{p^n}\setminus\{0,-1\}$, and therefore \begin{equation*}
    (y^{p^l}-y)^{p^{2l}+p^l}=0
\end{equation*} which gives $M_1=p^l-2$. Then \begin{equation*}
    \sum\limits_{u,v\in \mathbb{F}_{p^n}}S_{d_1}(u,v)^3=p^{2n}(3p^n-2+(p^n-1)(p^l-2))=p^{13l}+p^{3n}-p^{9l}.
\end{equation*}
\end{proof}

Based on  Theorem \ref{wf2 value}, Lemmas \ref{n_i(u,v)} and \ref{wf2 3powersum}, the value distribution of $S_{d_1}(u,v)$ is given as follows.

\begin{theorem}\label{wf2 value distribution} Let $p$ is a prime number and $p^l\equiv 2\,\,({\rm mod}\,\,3)$. Let $d_1=p^{3l}-p^{2l}+p^l$. Then the value distribution of the exponential sum $S_{d_1}(u,v)$ with $(u,v) \in {\mathbb{F}_{p^n}^*} \times {\mathbb{F}_{p^n}}$ can be described
 as follows.
	\begin{eqnarray*}
S_{d_1}(u,v)=
\left\{\begin{array}{llll}
-2p^{2l},&\mbox{occurs}&
({p^{8l}-3p^{7l}+3p^{6l}-p^{5l}-p^{4l}+3p^{3l}-3p^{2l}+p^l})/{6}&\mbox{times},\\
  -p^{2l},& \mbox{occurs}&
  p^{7l}-p^{6l}-p^{3l}+p^{2l}&\mbox{times},\\
  0,&\mbox{occurs}&
  ({p^{8l}-p^{7l}+p^{6l}-p^{5l}-3p^{4l}+p^{3l}-p^{2l}+p^l+2})/{2}&\mbox{times},\\
  p^{2l},& \mbox{occurs}&
  ({p^{8l} - p^{5l} - p^{4l}+p^l})/{3}&\mbox{times},\\
  p^{3l}-p^{2l},& \mbox{occurs}& p^l(p^{4l}-1)&\mbox{times},\\
  p^{3l},& \mbox{occurs}& p^{4l}-1&\mbox{times}.
  \end{array}\right.
  	\end{eqnarray*}
\end{theorem}
\begin{proof}
When $(u,v)$ runs through ${\mathbb{F}_{p^n}^*} \times {\mathbb{F}_{p^n}}$, the possible values of $S_{d}(u,v)$ are confirmed in Theorem \ref{wf2 value}. Suppose $-2p^{2l}$ occurs $E_0$ times, $-p^{2l}$ occurs $E_1$ times, $0$ occurs $E_2$ times, $p^{2l}$ occurs $E_3$ times, $p^{3l}-p^{2l}$ occurs $E_4$ times, $p^{3l}$ occurs $E_5$ times.
Then firstly we have
\begin{equation}\label{wf2 sum 1}
  E_0+E_1+E_2+E_3+E_4+E_5=p^n(p^n-1).
\end{equation}
Secondly, by Lemma \ref{wf2 3powersum} we can derive
\begin{equation}\label{wf2 sum 2}
     -2E_0-E_1+E_3+(p^l-1)\cdotp E_4+p^l\cdotp E_5=p^{6l}-p^{2l}
\end{equation}
\begin{equation}\label{wf2 sum 3}
     4E_0+E_1+E_3+(p^l-1)^2\cdotp E_4+p^{2l}\cdotp E_5=p^{2n}-p^n
\end{equation}
\begin{equation}\label{wf2 sum 4}
     -8E_0-E_1+E_3+(p^l-1)^3\cdotp E_4+p^{3l}\cdotp E_5=p^{7l}-p^{3l}.
\end{equation}
From Lemmas \ref{1 alpha}, \ref{n_i(u,v)} and Theorem \ref{wf2 value}, we can derive
\begin{equation}\label{wf2 sum 5}
  E_4=p^l(p^{4l}-1)
\end{equation}
\begin{equation}\label{wf2 sum 6}
  E_5=(p^{4l}-1).
\end{equation}
Solving the system equations consisting of (\ref{wf2 sum 1})-(\ref{wf2 sum 6}) we obtain the value distribution of $S_{d_1}(u,v)$ for $(u,v) \in {\mathbb{F}_{p^n}^*} \times {\mathbb{F}_{p^n}}$.
\end{proof}

Let $n=4l$ and $p^{l} \equiv 2\,\,({\rm mod}\,\,3)$.
Let $h_1(x)$ and $h_2(x)$ be the minimal polynomials of $\psi^{-1}$ and $\psi^{-({p^{3l}-p^{2l}+p^{l}})}$ over $\mathbb{F}_{p}$, respectively. Let $\mathcal{C}$ be the cyclic code with parity-check polynomial $h_1(x)h_2(x)$. By Delsarte's Theorem [\ref{Delsarte theorem}], the cyclic code $\mathcal{C}$ can be expressed as
\begin{equation*}
    \mathcal{C}=\left\{c_{u,v}=\left(\Tr ^{n}_1(u\psi^{i(p^{3l}-p^{2l}+p^{l})}+v\psi^i)\right)_{i=0}^{p^n-2}\mid u,v\in \mathbb{F}_{p^{n}} \right\}.
\end{equation*}
Denote code $\mathcal{C'}$ by
\begin{equation*}
    \mathcal{C'}=\left\{c'_{u,v}=\left(\Tr ^{n}_1(u\psi^{j(p^{3l}-p^{2l}+p^{l})}+v\psi^j)\right)_{j=0}^{\frac{p^n-p}{p-1}}\mid u,v\in \mathbb{F}_{p^{n}} \right\}.
\end{equation*}

Based on the result in Theorem \ref{wf2 value distribution}, the weight distribution of $\mathcal{C'}$ is determined as follows.
\begin{theorem}\label{wf2 code}
The $\mathcal{C'}$ is a $\left[\frac{p^n-1}{p-1},2n,p^{n-1}-p^{3l-1}\right]$ cyclic code. Furthermore, the weight distribution of $\mathcal{C'}$ is given in Table \ref{table5}.
\end{theorem}
\begin{table}[!ht]
    \centering
    \caption{The weight distribution of $\mathcal{C'}$\\}
    \label{table5}
    \begin{tabular}{lll}
      \toprule
      Weight              &  & Number of codewords         \\
      \midrule
      $0$  & & $1$   \\
$p^{n-1}-p^{3l-1}$ & & $p^{n}-1$ \\
$p^{n-1}-p^{3l-1}+p^{2l-1}$ & & $p^{l}(p^{n}-1)$ \\
$p^{n-1}-p^{2l-1}$ & &  $({p^{2n} - p^{5l} - p^{n}+p^{l}})/{3}$ \\
$p^{n-1}$ & & $({p^{2n}-p^{7l}+p^{6l}-p^{5l}-p^{n}+p^{3l}-p^{2l}+p^{l}})/{2}$ \\
$p^{n-1}+p^{2l-1}$ &  & $p^{7l}-p^{6l}-p^{3l}+p^{2l}$ \\
$p^{n-1}+2p^{2l-1}$ &  & $({p^{2n}-3p^{7l}+3p^{6l}-p^{5l}-p^{n}+3p^{3l}-3p^{2l}+p^{l}})/{6}$ \\
\bottomrule
    \end{tabular}
  \end{table}
\begin{proof}
    It is obvious that dimension of $\mathcal{C}$ is $2n$. For the codeword $c_{u,v}$ of $\mathcal{C}$, we also calculate its Hamming weight by using exponential sum:
\begin{equation*}
    \begin{aligned}
        \omega_H(c_{u,v}) &=p^{n}-1-\# \left\{x\in\mathbb{F}_{p^{n}}^*\mid \Tr_1^{n}(ux^{{p^{3l}-p^{2l}+p^{l}}}+vx)=0 \right\}\\
        &=p^{n}-\frac{1}{p}\sum\limits_{y\in \mathbb{F}_{p}}\sum\limits_{x\in \mathbb{F}_{p^{n}}}\omega_p^{y{\Tr ^{n}_1(ux^{p^{3l}-p^{2l}+p^{l}}+vx)}}\\
        &=p^{n-1}(p-1)-\frac{1}{p}\sum\limits_{y\in \mathbb{F}_{p}^*}\sum\limits_{x\in \mathbb{F}_{p^{n}}}\omega_p^{\Tr ^{n}_1(u{(xy)}^{{p^{3l}-p^{2l}+p^{l}}}+vxy)}\\
        &=p^{n-1}(p-1)-\frac{p-1}{p}\sum\limits_{x\in \mathbb{F}_{p^{n}}}\omega_p^{\Tr ^{n}_1(u{x}^{{p^{3l}-p^{2l}+p^{l}}}+vx)}.\\
    \end{aligned}
\end{equation*}
Therefore, the weight distribution of cyclic code $\mathcal{C}$ is directly determined by the value distribution of the exponential sum $\sum\limits_{x\in \mathbb{F}_{p^{n}}}\omega_p^{\Tr ^{n}_1(ux^{{p^{3l}-p^{2l}+p^{l}}}+vx)}$. If $u=0$, we have
\begin{equation*}
    \omega_H(c_{0,v})=
	\left\{ \begin{array}{llll}
		0&\mbox{if}&v=0,\\
		p^{n-1}(p-1)&\mbox{if}&v\neq0.
	\end{array}\right.
\end{equation*}
If $u\neq0$, the value distribution of $\omega_H(c_{u,v})$ follows from Theorem \ref{wf2 value distribution}. Let $\beta=\psi^{\frac{p^n-1}{p-1}}$ be a primitive element in $\mathbb{F}_p$. Note that
\begin{equation*}
  \Tr ^{n}_1\left(u\psi^{\big(j+k\frac{p^n-1}{p-1}\big)({p^{3l}-p^{2l}+p^{l}})}
  +v\psi^{j+k\frac{p^n-1}{p-1}}\right)=\beta^k\cdot\Tr ^{n}_1\left(u\psi^{j({p^{3l}-p^{2l}+p^{l}})}+v\psi^j\right),\mbox{ }0\leq k\leq p-2
\end{equation*}
implies
\begin{equation*}
  c_{u,v}=(c'_{u,v},\beta c'_{u,v},\cdots,\beta^{p-2}c'_{u,v}),
\end{equation*}
which yields the weight distribution of $\mathcal{C'}$.
\end{proof}

\begin{remark}
$\mathcal{C'}$ is $\beta^{-1}$-constacyclic code.
\end{remark}

In what follows, we provide some examples to verify our results.

\begin{example} Let $p=5$, $n=4$. Then $\mathcal{C'}$ is a $5$-ary cyclic code with parameters $[156,8,100]$. The weight enumerator of $\mathcal{C'}$ is
\begin{equation*}
1+624z^{100}+3120z^{105}+128960z^{120}+162240z^{125}+62400z^{130}+33280z^{135}.
\end{equation*}
\end{example}

\begin{example} Let $p=2$, $n=12$. Then $\mathcal{C'}$ is a binary cyclic code with parameters $[4095,24,1792]$. The weight enumerator of $\mathcal{C'}$ is
\begin{equation*}
1+4095z^{1792}+32760z^{1824}+5580120z^{2016}+7452900z^{2048}+1834560z^{2080}+1872780z^{2112}.
\end{equation*}
\end{example}

\begin{example} Let $p=11$, $n=4$. Then $\mathcal{C'}$ is a $11$-ary cyclic code with parameters $[1464,8,1210]$. The weight enumerator of $\mathcal{C'}$ is
\begin{equation*}
1+14640z^{1210}+161040z^{1221}+71394400z^{1320}+98234400z^{1331}+17714400z^{1342}+26840000z^{1353}.
\end{equation*}
\end{example}

\section{Conclusion}\label{conclusion}

\quad\; In this paper, we completely calculate the differential spectrum of $x^{p^{2l}-p^{l}+1}$. We investigate the $c$-differential uniformity of the power mapping $x^{p^{2l}-p^{l}+1}$. We also calculate the value distribution of a class of exponential sum. In addition,  we obtain a class of six-weight $p$-ary $\beta^{-1}$-constacyclic cyclic codes whose weight distribution can be explicitly determined.

\end{document}